\documentclass[12pt]{article}
\usepackage[cp1250]{inputenc}
\usepackage[verbose,a4paper,tmargin=0.5in,lmargin=1in,rmargin=1in]{geometry}
\usepackage[pdftex]{graphicx}
\usepackage{amsfonts}
\usepackage{indentfirst}
\usepackage{latexsym}
\usepackage{amsmath}
\usepackage{amsthm}
\usepackage{amssymb}
\usepackage{psfrag}
\usepackage[mathscr]{eucal} 
\usepackage{color}
\usepackage{cite}
\usepackage{longtable}
\usepackage{verbatim}
\usepackage{enumitem}
\usepackage[font=normalsize]{caption}
\usepackage{bigints}

\usepackage{graphicx}
\usepackage[all]{xy}
\usepackage{float}
\usepackage{afterpage}

\usepackage{physics}

\newtheorem{Twierdzenie}{Theorem}[section]

\newtheorem{Lemat}{Lemma}[section]
\newtheorem{Uwaga}{Remark}[section]

\newfloat{Scheme}{tbp}{muj}

\title{On a certain class of para-Hermite Einstein spaces.}

\author{$\textrm{Adam Chudecki}^{a,b,c}$}

\begin{document}
\sloppy

\maketitle

$^a$ Center of Mathematics and Physics, Lodz University of Technology, Al. Politechniki 11, 90-924 Lodz, Poland. 
\newline
$^b$ Institute of Physics, Faculty of Technical Physics, Information Technology and Applied Mathematics, Lodz University of Technology, Wólcza\'nska 217/221, 93-005 Lodz, Poland.
\newline
$^c$ adam.chudecki@p.lodz.pl
\newline
\newline
\textbf{Abstract}. 
A special class of (complex) para-Hermite Einstein spaces is analyzed. It is well-known that the self-dual Weyl tensor in para-Hermite Einstein spaces is of the Petrov-Penrose type [D]. In what follows we assume that the anti-self-dual Weyl tensor is algebraically degenerate. It is equivalent to the existence of an anti-self-dual congruence of null strings which is assumed not to be parallely propagated. Hence, spaces analyzed here are not Walker spaces. A classification of such spaces is given and the explicit metrics are found.
\newline
\newline
\textbf{PACS numbers:} 04.20.Jb, 04.20.Gz.
\newline
\newline
\textbf{Key words:} para-Hermite Einstein spaces, congruences of null strings.


\section{Introduction}
\setcounter{equation}{0}

The paper is devoted to a special class of para-Hermite Einstein (pHE) spaces. It is a continuation of the paper \cite{Chudecki_Ref_3}. It is also a part of more extensive research programme devoted to algebraically degenerate para-K\"ahler Einstein spaces (pKE-spaces) and para-Hermite Einstein spaces (pHE-spaces).

PKE-spaces and pHE-spaces are equipped with a pair of congruences of null strings of the same duality\footnote{\textsl{A congruence of self-dual (SD) null strings} is a family of totally null, SD, 2-dimensional surfaces. In Penrose terminology such structures are called $\alpha$-surfaces. An anti-self-dual (ASD) congruence of null strings is defined analogously and it corresponds to Penrose $\beta$-surfaces. In this text a congruence of null strings will be abbreviated by the symbol $\mathcal{C}$.}. If an orientation is chosen in such a manner that these $\mathcal{C}s$ are SD then the SD Weyl tensor is type-[D]. The difference between pKE-spaces and pHE-spaces lies in properties of these $\mathcal{C}s$. In pKE-spaces both $\mathcal{C}s$ are parallely propagated while in pHE-spaces they are not parallely propagated. Hence, pKE-spaces are spaces of the types\footnote{See Ref. \cite{Chudecki_Ref_1,Chudecki_Ref_2, Chudecki_Ref_3} for an explanation of a formalism, symbols and abbreviations.} $[\textrm{D}]^{nn} \otimes [\textrm{any}]$ while pHE-spaces are spaces of the types $[\textrm{D}]^{ee} \otimes [\textrm{any}]$. Despite this difference, both pKE-spaces and pHE-spaces belong to a more general class of \textsl{hyperheavenly spaces ($\mathcal{HH}$-spaces)} \cite{Plebanski_Robinson_2}.

$\mathcal{HH}$-spaces are complex 4-dimensional Einstein spaces such that the SD Weyl tensor is algebraically degenerate. They have amazing property: vacuum Einstein field equations with cosmological constant in $\mathcal{HH}$-spaces were reduced to a single equation which is called \textsl{the hyperheavenly equation ($\mathcal{HH}$-equation)}. A solution of the $\mathcal{HH}$-equation gives a potential (which is called \textsl{the key function}) which completely determines the metric. This outstanding fact justifies why the $\mathcal{HH}$-spaces theory is a brilliant tool in various problems of 4D complex and real Einstein geometries. $\mathcal{HH}$-spaces are especially  useful in 4D real neutral geometries because real neutral slice of a complex $\mathcal{HH}$-metric can be easily obtained. The results published in \cite{Chudecki_Ref_1,Chudecki_Ref_2,Chudecki_Ref_3} proved that an approach to pKE-spaces and pHE-spaces via $\mathcal{HH}$-spaces was a "bull's eye".

PKE-spaces and pHE-spaces with algebraically general ASD Weyl tensor are problematic. According to our best knowledge, explicit examples of such spaces have not been found yet. The issue changes if one assumes an algebraic degeneracy of the ASD Weyl tensor. In such a case pKE-spaces were completely solved in all generality \cite{Chudecki_Ref_1,Chudecki_Ref_2}. PHE-spaces are more complicated but a significant progress in the subject has been done in \cite{Chudecki_Ref_3} where all algebraically degenerate pHE-spaces equipped with a parallely propagated ASD $\mathcal{C}$ have been found (a congruence of null strings which is parallely propagated is called \textsl{nonexpanding}; spaces which are equipped with nonexpanding $\mathcal{C}$ are called \textsl{Walker spaces}). Hence, all Einstein metrics of the types $[\textrm{D}]^{ee} \otimes [\textrm{deg}]^{n}$ are known. The current paper is devoted to pHE-spaces for which an ASD $\mathcal{C}$ is not parallely propagated (i.e., it is \textsl{expanding}). In other words, we deal here with spaces of the types $[\textrm{D}]^{ee} \otimes [\textrm{deg}]^{e}$ which are not Walker anymore.

To be more precise, we focus our attention on a special class of type-$[\textrm{D}]^{ee} \otimes [\textrm{deg}]^{e}$ pHE-spaces. It is known that SD and ASD $\mathcal{C}s$ intersect and such an intersection constitutes \textsl{a congruence of null geodesics}\footnote{A congruence of null geodesics will be abbreviated by the symbol $\mathcal{I}$.}, characterized by \textsl{an expansion} and \textsl{a twist}. All algebraically degenerate pHE-spaces are equipped with at least two $\mathcal{I}s$. If the ASD Weyl tensor is type-[D] then there are four $\mathcal{I}s$. In what follows we assume that a twist of an arbitrarily chosen $\mathcal{I}$ vanishes and we will prove that this assumption together with vacuum Einstein field equations imply (perhaps surprisingly) that all other $\mathcal{I}s$ are nontwisting. Hence, the ASD Weyl tensor is restricted to the types [II] or [D]. Consequently, the current paper is devoted to spaces of the types $[\textrm{D}]^{ee} \otimes [\textrm{II}]^{e}$ and $[\textrm{D}]^{ee} \otimes [\textrm{D}]^{ee}$ with all $\mathcal{I}s$ being nontwisting.

Our paper is organized, as follows. In Section \ref{section_HH_spaces} a brief introduction to type-[D] $\mathcal{HH}$-spaces is presented. The general formulas for a metric, conformal curvature, gauge freedom and symmetries are given and they are specialized to the twist-free types $[\textrm{D}]^{ee} \otimes [\textrm{deg}]^{e}$. It is shown that the only possible types are $[\textrm{D}]^{ee} \otimes [\textrm{II}]^{e}$ and $[\textrm{D}]^{ee} \otimes [\textrm{D}]^{ee}$. Sections \ref{section_Dee_x_IIe} and \ref{section_Dee_x_Dee} are devoted to general solutions and to solutions with additional symmetries. 

We will close this Section with an important remark. The current paper is thought as a continuation of \cite{Chudecki_Ref_3}. Thus, Section \ref{section_HH_spaces} devoted to preliminary remarks has been shortened as much as possible. We hope that the Reader of this paper will be also interested in results presented in \cite{Chudecki_Ref_3} as well as in \cite{Chudecki_Ref_1,Chudecki_Ref_2}. These references contain a comprehensive introduction to the spinorial formalism in Infeld - Van der Waerden - Pleba\'nski notation, hyperheavenly spaces, congruences of null strings, congruences of null geodesics and - what is especially important - abbreviations and symbols which we use. It is warmly suggested to the Reader of this paper to treat Ref. \cite{Chudecki_Ref_3} as a mandatory textbook to the subject.

All considerations are local. Functions are holomorphic and coordinates are complex. Hence, the metrics (\ref{metryka_TypII_pm_pm_ostateczna}), (\ref{metryka_TypII_pm_mm_ostateczna}), (\ref{metryka_TypD_ostateczna_pmpmpmpm}) and (\ref{metryka_TypD_ostateczna_pmmmmmpm}) are in general complex holomorphic. All these metrics have real neutral slices. A procedure of obtaining such slices is more then straightforward: the coordinates should be treated as real ones and the functions as smooth ones.


\renewcommand{\arraystretch}{1.5}
\setlength\arraycolsep{2pt}
\setcounter{equation}{0}

\section{Hyperheavenly spaces}
\label{section_HH_spaces}
\setcounter{equation}{0}

\subsection{Hyperheavenly spaces of the types $[\textrm{D}]^{ee} \otimes [\textrm{any}]$}

\subsubsection{The metric}
\label{subsubsekcja_metric}

The metric of an $\mathcal{HH}$-space of the type $[\textrm{D}]^{ee} \otimes [\textrm{any}]$ has the form
\begin{equation}
\label{metryka_HH_ekspandujaca_D_any}
\frac{1}{2} ds^{2} = x^{-2} \left( dqdy-dpdx + \mathcal{A} \, dp^{2} - 2 \mathcal{Q} \, dpdq + \mathcal{B} \, dq^{2} \right)  = e^{1}e^{2} + e^{3}e^{4}
\end{equation}
where $(q,p,x,y)$ are local coordinates called \textsl{the hyperheavenly coordinates}. A null tetrad $(e^{1},e^{2},e^{3},e^{4})$ is called \textsl{Pleba\'nski tetrad} and it reads
\begin{equation}
\label{tetrada_Plebanskiego__HH_ekspandujaca}
e^{3} = x^{-2} dp, \ e^{1} = -x^{-2}  dq, \ e^{4} = -dx + \mathcal{A} \, dp - \mathcal{Q} \, dq,\ e^{2} = -dy + \mathcal{Q} \, dp - \mathcal{B} \, dq
\end{equation}
Functions $\mathcal{A} = \mathcal{A} (q,x,y)$, $\mathcal{Q} = \mathcal{Q} (q,x,y)$ and $\mathcal{B} = \mathcal{B} (q,x,y)$ stand for
\begin{subequations}
\label{zwiazek_miedzy_Qab_i_W}
\begin{eqnarray}
\label{zwiazek_miedzy_w_i_A}
\mathcal{A} &:=& -x W_{yy} + \mu_{0} x^{3} + \frac{\Lambda}{6}
\\
\label{zwiazek_miedzy_w_i_Q}
\mathcal{Q} &:=& xW_{xy}-W_{y}
\\ 
\label{zwiazek_miedzy_w_i_B}
\mathcal{B} &:=& -xW_{xx} +2W_{x}
\end{eqnarray}
\end{subequations}
where $\mu_{0} =1$ and $\Lambda$ is the cosmological constant\footnote{The constant $\mu_{0}$ can be brought to an arbitrary value, in particular but not necessarily, to $1$. Thus, in the rest of the paper we will be using the symbol $\mu_{0}$.}. \textsl{The key function} $W=W(q,x,y)$ is a holomorphic function which satisfies \textsl{the expanding hyperheavenly equation}
\begin{eqnarray}
\label{HH_equation_ogolne}
W_{xx}W_{yy} - W_{xy}^{2} + \frac{2}{x} (W_{y}W_{xy} - W_{x}W_{yy}) +\frac{1}{x} W_{qy} &&
\\ \nonumber
 - \mu_{0} (x^{2} W_{xx} - 3x W_{x} + 3W) - \frac{\Lambda}{6x} W_{xx}&=&0
\end{eqnarray}
The conformal curvature is given by the SD and ASD conformal curvature coefficients. The only nonzero SD conformal curvature coefficient read
\begin{equation}
C^{(3)} = -2 \mu_{0} x^{3}
\end{equation}
and for ASD ones we find
\begin{eqnarray}
\label{ASD_conformal_curvature}
\frac{1}{2x^{3}} \dot{C}^{(5)} &=&  \partial_{x}^{4} \left( W - \frac{\mu_{0}}{4} x^{2} y^{2} \right)
\\ \nonumber
\frac{1}{2x^{3}} \dot{C}^{(4)} &=&  \partial_{x}^{3}\partial_{y} \left( W - \frac{\mu_{0}}{4} x^{2} y^{2} \right)
\\ \nonumber
\frac{1}{2x^{3}} \dot{C}^{(3)} &=&  \partial_{x}^{2}\partial_{y}^{2} \left( W - \frac{\mu_{0}}{4} x^{2} y^{2} \right)
\\ \nonumber
\frac{1}{2x^{3}} \dot{C}^{(2)} &=&  \partial_{x}\partial_{y}^{3} \left( W - \frac{\mu_{0}}{4} x^{2} y^{2} \right)
\\ \nonumber
\frac{1}{2x^{3}} \dot{C}^{(1)} &=& \partial_{y}^{4} \left( W - \frac{\mu_{0}}{4} x^{2} y^{2} \right)
\end{eqnarray}

\subsubsection{Gauge freedom}

Transformations of the hyperheavenly coordinates which leave the metric (\ref{metryka_HH_ekspandujaca_D_any}) invariant read
\begin{equation}
\label{gauge}
q'=q'(q), \ p'= p + h(q), \ x'= x, \ y' = \frac{1}{f} ( y +  h_{q} x) + \sigma (q)
\end{equation}
where $h=h(q)$, $\sigma=\sigma(q)$ and $f=f(q) := \dfrac{dq'}{dq}$ are arbitrary gauge functions. Under (\ref{gauge}) the key function $W$ transforms as follows
\begin{eqnarray}
\label{gauge_for_key_function}
f^{2}  W' &=& W + \frac{1}{2} \mu_{0} h_{q} \left( x^{3}y + \frac{1}{2}  h_{q} x^{4} \right) - \frac{1}{3} L \, x^{3} + \frac{f_{q}}{2f} \, xy
\\ \nonumber
&&   - \frac{1}{2}   f \frac{\partial}{\partial q} \left( \frac{h_{q}}{f} \right)   x^{2}
   -\frac{\Lambda}{6}  h_{q} \,  y - \left( \frac{1}{2}  f \sigma_{q} + \frac{\Lambda}{12} h_{q}^{2}  \right)   x - M
\end{eqnarray}
where $L=L(q)$ and $M=M(q)$ are arbitrary gauge functions such that
\begin{equation}
\label{zwiazek_miedzy_L_i_M}
3\mu_{0} M - f^{\frac{1}{2}} \partial_{q}^{2} (f^{-\frac{1}{2}}) + \frac{\Lambda}{3} L = 0
\end{equation}

\subsubsection{Congruences of SD null strings}

$\mathcal{HH}$-spaces of the types $[\textrm{D}]^{ee} \otimes [\textrm{any}]$ are equipped with two complementary congruences of null strings. We denote these congruences by $\mathcal{C}_{m^{A}}$ and $\mathcal{C}_{n^{A}}$. Spinors which generate these $\mathcal{C}s$ and their expansion read
\begin{subequations}
\label{ekspansja_CmA_and_CnA}
\begin{eqnarray}
\label{ekspansja_CmA}
\mathcal{C}_{m^{A}}: && m_{A}=[0,m], m \ne 0; \ \  M_{\dot{A}} = -\sqrt{2} \, \frac{m}{x} \, [1,0]
\\ 
\label{ekspansja_CnA}
\mathcal{C}_{n^{A}}: && n_{A}=[n,0], n \ne 0; \ \  N_{\dot{A}} = \sqrt{2} \, nx \, [-\mathcal{Q}, \mathcal{A}]
\end{eqnarray}
\end{subequations}
Because $M_{\dot{A}} \ne 0$ and $N_{\dot{A}} \ne 0$ hold true both $\mathcal{C}_{m^{A}}$ and $\mathcal{C}_{n^{A}}$ are expanding and a degeneration to a nonexpanding case is not admitted.

\subsection{Hyperheavenly spaces of the types $[\textrm{D}]^{ee} \otimes [\textrm{deg}]^{e}$}

\subsubsection{Congruence $\mathcal{C}_{m^{\dot{A}}}$}
\label{ASD_congruence_section}

An algebraic degeneracy of the ASD Weyl tensor is equivalent to the existence of an ASD $\mathcal{C}$. We denote this congruence by $\mathcal{C}_{m^{\dot{A}}}$ (i.e., it is generated by a dotted spinor $m^{\dot{A}}$) and its expansion we denote by $M_{A}$. The spinor $m^{\dot{A}}$ can be always scaled to the form  $m_{\dot{A}} \sim [z,1]$, $z=z(q,x,y)$. The ASD null string equations read
\begin{subequations}
\label{kongruencja_mdotA_rownania}
\begin{eqnarray}
\label{kongruencja_mdotA_rownania_1}
&& z_{x} - zz_{y}=0
\\ 
\label{kongruencja_mdotA_rownania_2}
&& z_{q}  - z_{y} \mathcal{Z} + z \frac{\partial \mathcal{Z}}{\partial y} - \frac{\partial \mathcal{Z}}{\partial x} =0 , \ \mathcal{Z} := \mathcal{B} + 2z \mathcal{Q} + z^{2} \mathcal{A}
\end{eqnarray}
\end{subequations}
and the expansion $M_{A}$ has the form
\begin{subequations}
\label{ekspansja_pierwszej_ASD_struny}
\begin{eqnarray}
\label{ekspansja_pierwszej_ASD_struny_1}
\frac{x}{\sqrt{2}} M_{1} &=& -x z_{y}  - 1
\\ 
\label{ekspansja_pierwszej_ASD_struny_2}
\frac{1}{\sqrt{2}x}  M_{2} &=&   x z \frac{\partial}{\partial y} (\mathcal{Q} +z \mathcal{A}) - x \frac{\partial}{\partial x} (\mathcal{Q} +z \mathcal{A}) + (1 - x z_{y}) (\mathcal{Q} +z \mathcal{A})   \ \ \ \ \ \ \ 
\end{eqnarray}
\end{subequations}
The transformation formula for $z$ reads
\begin{equation}
\label{transformacja_na_z_ina_w}
f z' = z - h_{q}
\end{equation}

\subsubsection{Congruences of null geodesics}
\label{ASD_congruence_null_geodesics_section}

An $\mathcal{HH}$-space of the type $[\textrm{D}]^{ee} \otimes [\textrm{deg}]^{e}$ is automatically equipped with two different congruences of null geodesics. We denote them by $\mathcal{I}_{1} = \mathcal{I} (\mathcal{C}_{m^{A}}, \mathcal{C}_{m^{\dot{A}}})$ and $\mathcal{I}_{3} = \mathcal{I} (\mathcal{C}_{n^{A}}, \mathcal{C}_{m^{\dot{A}}})$ (see Figure \ref{Congruences}). Expansions and twists of these $\mathcal{I}s$ read
\begin{eqnarray}
\label{wlasnosci_przeciec}
&& \theta_{1} \sim xz_{y} +2, \ \varrho_{1} \sim z_{y}
\\ \nonumber
&& \theta_{3} \sim -n M_{2} - zN_{\dot{2}} + N_{\dot{1}}, \ \varrho_{3} \sim -n M_{2} + zN_{\dot{2}} - N_{\dot{1}}
\end{eqnarray}

\begin{figure}[ht]
\begin{center}
\includegraphics[scale=0.9]{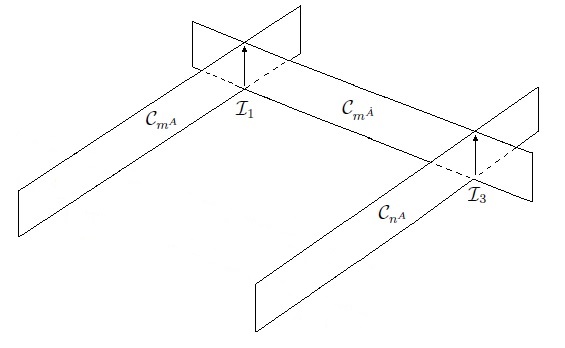}
\caption {Congruences of null strings and congruences of null geodesics in algebraically degenerate para-Hermite Einstein spaces.}
\label{Congruences}
\end{center}
\end{figure}

\subsection{Hyperheavenly spaces of the types $\{  [\textrm{D}]^{ee} \otimes [\textrm{deg}]^{e} ,[+-,??]  \}$}

In what follows we assume that one of the twists vanishes. It is arbitrary which twist we put zero but it is much more convenient to put
\begin{equation}
\label{twist_varrho1_zero}
\varrho_{1} = 0
\end{equation}
\begin{Uwaga}
\normalfont
\label{Uwaga_o_calkowalnosci}
Note, that vanishing of the twist $\varrho_{1}$ of $\mathcal{I}_{1}$ has a deep geometrical interpretation. Leafs of the congruence $\mathcal{C}_{m^{A}}$ are integral manifolds of a 2D-distribution $\{ m_{A} m_{\dot{B}}, m_{A} x_{\dot{B}} \}$, $m_{\dot{B}} x^{\dot{B}} \ne 0$. Analogously, leafs of the congruence $\mathcal{C}_{m^{\dot{A}}}$ are integral manifolds of a 2D-distribution $\{ m_{A} m_{\dot{B}}, y_{A} m_{\dot{B}} \}$, $m_{A} y^{A} \ne 0$. If we construct a 3D-distribution spanned by the vectors $\{ m_{A} m_{\dot{B}}, m_{A} x_{\dot{B}}, y_{A} m_{\dot{B}} \}$ then the condition (\ref{twist_varrho1_zero}) is equivalent to the fact that the 3D-distribution is integrable.
\end{Uwaga}
Eqs. (\ref{twist_varrho1_zero}) and (\ref{kongruencja_mdotA_rownania_1}) imply that $z_{x}=z_{y}=0 \ \Longrightarrow \ z=z(q)$. But if it is so, then from (\ref{transformacja_na_z_ina_w}) we find that $z$ can be gauged away without any loss of generality. Hence, in the rest of the paper we put $z=0$.
\begin{Uwaga}
\normalfont
The choice $z=0$ leaves us with the gauge (\ref{gauge}) restricted to the condition $h=h_{0} =\textrm{const}$.
\end{Uwaga}

With $z=0$ formulas simplify considerably. Indeed, expansions of $\mathcal{C}_{m^{\dot{A}}}$ take the form
\begin{subequations}
\label{ekspansja_pierwszej_ASD_struny_po_zmianach}
\begin{eqnarray}
\label{ekspansja_pierwszej_ASD_struny_1_po_zmianach}
\frac{x}{\sqrt{2}} M_{1} &=&   - 1
\\ 
\label{ekspansja_pierwszej_ASD_struny_2_po_zmianach}
\frac{1}{\sqrt{2}x}  M_{2} &=&    \mathcal{Q}   - x \mathcal{Q}_{x}
\end{eqnarray}
\end{subequations}
Because $M_{1} \ne 0$, the congruence $\mathcal{C}_{m^{\dot{A}}}$ is always expanding. Optical parameters for $\mathcal{I}s$ read
\begin{subequations}
\begin{eqnarray}
\label{wlasnosci_przeciec_po_zmianach}
&& \theta_{1} \sim 2, \ \varrho_{1}=0
\\ 
\label{wlasnosci_przeciec_po_zmianach_I3}
&& \theta_{3} \sim 2 \mathcal{Q} - x \mathcal{Q}_{x}, \ \varrho_{3} \sim \mathcal{Q}_{x}
\end{eqnarray}
\end{subequations}
Thus, $\mathcal{I}_{1}$ is always expanding and nontwisting, $p(\mathcal{I}_{1})=[+-]$. Properties of $\mathcal{I}_{3}$ 
depend on $\mathcal{Q}$ and they will be determined when we find the form of the key function. From the SD null string equation (\ref{kongruencja_mdotA_rownania_2}) with $z=0$ we find that $\mathcal{B}_{x}=0$. Using (\ref{zwiazek_miedzy_w_i_B}) we arrive at the key function in the form
\begin{equation}
\label{funkcja_kluczowa_postac_wyjsciowa}
W= B(q,y) \, x^{3} + A(q,y) \, x + C(q,y)
\end{equation}
Under (\ref{gauge}) the functions $A$, $B$ and $C$ transform as follows
\begin{subequations}
\begin{eqnarray}
\label{transformacja_A}
f^{2} B' &=& B - \frac{1}{3}L
\\ 
\label{transformacja_B}
f^{2} A' &=& A + \frac{f_{q}}{2f} y - \frac{1}{2} f \sigma_{q}
\\ 
\label{transformacja_C}
f^{2} C' &=& C-M
\end{eqnarray}
\end{subequations}
Feeding the $\mathcal{HH}$-equation (\ref{HH_equation_ogolne}) with (\ref{funkcja_kluczowa_postac_wyjsciowa}) one finds that $\mathcal{HH}$-equation becomes a 4th order polynomial in $x$. From the coefficient standing at $x^4$ one finds that $B_{y}=0$. Hence, $B=B(q)$ and from (\ref{transformacja_A}) it follows that $B$ can be gauged away without any loss of generality. 
\begin{Uwaga}
\normalfont
The choice $B=0$ leaves us with the gauge (\ref{gauge}) restricted to the condition $L=0$.
\end{Uwaga}
With $B=0$ the key function simplifies significantly and it reads
\begin{equation}
\label{funkcja_kluczowa_postac_wyjsciowa_uproszczona}
W=  A(q,y) \, x + C(q,y)
\end{equation}
while the $\mathcal{HH}$-equation reduces to a pair of equations
\begin{subequations}
\label{HH_resztki}
\begin{eqnarray}
\label{HH_resztki_1}
 A_{y}^{2} - 2A A_{yy} + A_{yq} - 3 \mu_{0} C &=& 0
\\ 
\label{HH_resztki_2}
C_{yq} - 2A C_{yy} + 2 A_{y} C_{y}   &=& 0
\end{eqnarray}
\end{subequations}

Now we are ready to determine the properties of $\mathcal{I}_{3}$. We find that $\theta_{3} \sim C_{y}$ and $\varrho_{3} =0$. Hence, we deal with two different classes:
\begin{subequations}
\begin{eqnarray}
\label{warunek_na_podtyp_pm_pm}
\textrm{Types} \ \  \{ [\textrm{D}]^{ee} \otimes [\textrm{deg}]^{e} , [+-,+-]   \} \ \ \textrm{for} \ \ C_{y} \ne 0
\\ 
\label{warunek_na_podtyp_pm_mm}
\textrm{Types} \ \  \{ [\textrm{D}]^{ee} \otimes [\textrm{deg}]^{e} , [+-,--]   \} \ \ \textrm{for} \ \ C_{y} = 0
\end{eqnarray}
\end{subequations}

Let us summarize the results of this Section. We have found that the metric of any $\mathcal{HH}$-space of the type $\{ [\textrm{D}]^{ee} \otimes [\textrm{deg}]^{e} , [+-,+-]   \}$ or $\{ [\textrm{D}]^{ee} \otimes [\textrm{deg}]^{e} , [+-,--]   \}$ takes the form
\begin{eqnarray}
\label{metryka_HH_Dee_x_dege_twistfree}
\frac{1}{2} ds^{2} &=& x^{-2} \bigg\{ dq dy - dp dx + \left( \mu_{0} \, x^{3} - A_{yy} \, x^{2} - C_{yy} \, x + \frac{\Lambda}{6}  \right) \, dp^{2} 
\\ \nonumber
&& \ \ \ \ \ \ \ \ + 2 C_{y} \, dpdq + 2A \, dq^{2} \bigg\}
\end{eqnarray}
where $\mu_{0}=1$, $\Lambda$ is the cosmological constant, $A=A(q,y)$ and $C=C(q,y)$ are holomorphic functions which satisfy Eqs. (\ref{HH_resztki}).

\subsubsection{ASD conformal curvature}

Before we focus on exact solutions, we discuss possible Petrov-Penrose types of the ASD Weyl spinor. Inserting (\ref{funkcja_kluczowa_postac_wyjsciowa_uproszczona}) into (\ref{ASD_conformal_curvature}) one finds that nonzero ASD conformal coefficients read 
\begin{eqnarray}
\label{ASD_conformal_curvature_1}
\dot{C}^{(1)} &=& 2x^{3} (A_{yyyy} \, x + C_{yyyy})
\\ \nonumber
\dot{C}^{(2)} &=& 2 x^{3}   A_{yyy}
\\ \nonumber
\dot{C}^{(3)} &=& -2 x^{3} \mu_{0}
\end{eqnarray}
The ASD Weyl spinor takes the form 
\begin{equation}
\label{spinorowa_postac_ASD_Weyl_1}
C_{\dot{A}\dot{B}\dot{C}\dot{D}} = -6 \mu_{0} x^{3} \, n^{+}_{(\dot{A}} n^{-}_{\dot{B}} k_{\dot{C}} k_{\dot{D})}
\end{equation}
where
\begin{equation}
n^{\pm}_{\dot{A}} = l_{\dot{A}} - \left( \frac{A_{yyy}}{3 \mu_{0}} \pm \frac{1}{6 \mu_{0}} \sqrt{4 A_{yyy}^{2} + 6 \mu_{0} (A_{yyyy} \, x + C_{yyyy})}  \right)  k_{\dot{A}}
\end{equation}
while $(l_{\dot{A}}, k_{\dot{B}})$ is a basis of 1-index dotted spinors normalized to 1, $k^{\dot{A}} l_{\dot{A}} = 1$. From (\ref{spinorowa_postac_ASD_Weyl_1}) it follows that the ASD Weyl spinor can be type-$[\textrm{II}]^{e}$ or type-$[\textrm{D}]^{ee}$. A stronger algebraic degeneration involves $\mu_{0} = 0$ what is a contradiction. Criteria are gathered in the Table \ref{Kryteria_krzywizy_ASD}.
\begin{table}[ht]
\begin{center}
\begin{tabular}{|c|c|}   \hline
Petrov-Penrose type &  Criteria     \\  \hline
$[\textrm{D}]^{ee} \otimes [\textrm{II}]^{e}$ & $A_{yyyy} \ne 0$ or $3 \mu_{0} C_{yyyy} + 2 A_{yyy}^{2} \ne 0 $  \\ \hline
$[\textrm{D}]^{ee} \otimes [\textrm{D}]^{ee}$ & $A_{yyyy} = 0$ and $3 \mu_{0} C_{yyyy} + 2 A_{yyy}^{2} = 0 $  \\ \hline
\end{tabular}
\caption{Criteria for Petrov-Penrose types.}
\label{Kryteria_krzywizy_ASD}
\end{center}
\end{table}

Note, that the spinor $k_{\dot{A}}$ is a double dotted Penrose spinor so it generates the ASD $\mathcal{C}$. Hence, it is proportional to $m_{\dot{A}}$.

\subsubsection{Symmetries}
\label{Symmetries_general_approach}

A homothetic vector admitted by the metric (\ref{metryka_HH_Dee_x_dege_twistfree}) has the form\footnote{To prove (\ref{homotetyczny_ogolny}) and other formulas buried in this Section it is enough to insert the form of the key function (\ref{funkcja_kluczowa_postac_wyjsciowa_uproszczona}) into formulas from Section 3.1.3 of  \cite{Chudecki_Ref_3}.}
\begin{equation}
\label{homotetyczny_ogolny}
K = \widetilde{a} \, \partial_{q} + \widetilde{c}_{0} \, \partial_{p} - (\widetilde{a}_{q} y + \widetilde{\epsilon} ) \partial_{y} + \frac{2}{3} \chi_{0} (2p \partial_{p} - x \partial_{x} + y \partial_{y} )
\end{equation}
where $\widetilde{a}=\widetilde{a}(q)$, $\widetilde{c_{0}}= \textrm{const}$ and $\widetilde{\epsilon} = \widetilde{\epsilon} (q)$, $\chi_{0}$ is a homothetic factor. Equations $\nabla_{(a} K_{b)} = \chi_{0} g_{ab}$  reduce to the system
\begin{subequations}
\label{zredukowane_master}
\begin{eqnarray}
\label{zredukowane_master_1}
&& \widetilde{a} A_{q} + \left( \frac{2}{3} \chi_{0} y - \widetilde{a}_{q} y - \widetilde{\epsilon} \right) A_{y} + \left( 2 \widetilde{a}_{q} - \frac{2}{3} \chi_{0} \right) A - \frac{1}{2} \widetilde{a}_{qq} y - \frac{1}{2} \widetilde{\epsilon}_{q} =0 
\\ 
\label{zredukowane_master_2}
&& \widetilde{a} C_{q} + \left( \frac{2}{3} \chi_{0} y - \widetilde{a}_{q} y - \widetilde{\epsilon} \right) C_{y} + 2 \widetilde{a}_{q} C - \frac{\widetilde{a}_{qqq}}{6 \mu_{0}} = 0
\end{eqnarray}
\end{subequations}
with an integrability condition 
\begin{equation}
\label{symmetry_integrability_condition}
\Lambda \chi_{0} =0
\end{equation}
Under (\ref{gauge}) the transformation formulas for $\widetilde{a}$, $\widetilde{c}_{0}$ and $\widetilde{\epsilon}$ read
\begin{subequations}
\label{symmetry_transformation}
\begin{eqnarray}
\label{symmetry_transformation_a}
\widetilde{a}' &=& f \widetilde{a}
\\ 
\label{symmetry_transformation_c}
\widetilde{c}_{0}' &=&  \widetilde{c}_{0}  - \frac{4}{3}\chi_{0} h_{0}
\\ 
\label{symmetry_transformation_epsilon}
\widetilde{\epsilon}' &=&  f^{-1} \widetilde{\epsilon} - \sigma \left(  \widetilde{a}_{q} - \frac{2}{3} \chi_{0} + \widetilde{a} \partial_{q} (\ln (\sigma f)) \right)
\end{eqnarray}
\end{subequations}
Eqs. (\ref{zredukowane_master}) are identically satisfied for $\widetilde{a} = \widetilde{\epsilon} = \chi_{0}=0$ what implies that every space of the type $\{ [\textrm{D}]^{ee} \otimes [\textrm{deg}]^{e}, [+-, +-] \}$ or $\{ [\textrm{D}]^{ee} \otimes [\textrm{deg}]^{e}, [+-, --] \}$ is equipped with a Killing vector\footnote{In fact, it holds true also for much more general spaces of the types $[\textrm{D}]^{ee} \otimes [\textrm{any}]$.}
\begin{equation}
\label{pierwszy_Killing}
K_{1} = \frac{\partial}{\partial p}
\end{equation}
It means that without any loss of generality one puts $\widetilde{c}_{0} = 0$ in (\ref{homotetyczny_ogolny}). Indeed, a linear combination of two homothetic vectors with constant coefficients is also a homothetic vector. Hence, any homothetic vector other then (\ref{pierwszy_Killing}) has the form (\ref{homotetyczny_ogolny}) with $\widetilde{c}_{0} = 0$. 
\begin{Uwaga}
\normalfont
The choice $\widetilde{c}_{0} = 0$ leaves us with the gauge (\ref{gauge}) restricted to the condition $h=0$.
\end{Uwaga}


\section{Types $[\textrm{D}]^{ee} \otimes [\textrm{II}]^{e}$}
\label{section_Dee_x_IIe}
\setcounter{equation}{0}

\subsection{Type $\{ [\textrm{D}]^{ee} \otimes [\textrm{II}]^{e}, [+-,+-] \}$}

In this Section we deal with the case for which
\begin{equation}
\label{warunki_Typ_II_pm_pm}
C_{y} \ne 0,  \ A_{yyyy} \ne 0 \ \textrm{or} \ 2A_{yyy}^{2} +3 \mu_{0} C_{yyyy} \ne 0
\end{equation}
(Compare (\ref{warunek_na_podtyp_pm_pm}) and Table \ref{Kryteria_krzywizy_ASD}).

\subsubsection{General case}

\begin{Twierdzenie}
\label{Twierdzenie_typ_II_pm_pm}
Let $(\mathcal{M}, ds^{2})$ be an Einstein complex space of the type $\{ [\textrm{D}]^{ee} \otimes [\textrm{II}]^{e}, [+-,+-] \}$. Then there exists a local coordinate system $(q,p,x,w)$ such that the metric takes the form 
\begin{eqnarray}
\label{metryka_TypII_pm_pm_ostateczna}
\frac{1}{2} ds^{2} &=& x^{-2} \bigg\{ -dpdx +\frac{M_{q}}{6 \mu_{0}} dq dw + \frac{M_{q}}{3 \mu_{0}} dpdq
\\ \nonumber
&& \ \ \ \ \ \
+ \left( \mu_{0} x^{3} + \frac{\Lambda}{6} - \frac{1}{2} M \, x^{2} -  \frac{M_{qw}}{M_{q}} \, x  \right) dp^{2} \bigg\}
\end{eqnarray}
where $\mu_{0}=1$, $\Lambda$ is the cosmological constant and 
\begin{equation}
\label{Typ_II_pm_pm_jedyne_rownanie}
36 \mu_{0} M_{w} = -M^{3} + a(w) \, M + b(w)
\end{equation}
where $M=M(q,w)$, $a=a(w)$ and $b=b(w)$ are holomorphic functions such that $M_{q} \ne 0$ and
\begin{equation}
\label{Typ_II_pm_pm_jedyne_warunki}
a_{w} M + b_{w} \ne 0 \ \ \textrm{or} \ \ 36 \mu_{0} a_{ww} - aa_{w} - 3a_{w} M^{2} - 6b_{w} M \ne 0
\end{equation}
\end{Twierdzenie}

\begin{proof}
The condition $C_{y} \ne 0$ allows to rewrite Eq. (\ref{HH_resztki_2}) in the form
\begin{equation}
\frac{\partial}{\partial q} \left( \frac{1}{C_{y}} \right) = \frac{\partial}{\partial y} \left( \frac{2A}{C_{y}} \right)
\end{equation}
Hence, one concludes the existence of a function $\Omega = \Omega (q,y)$ such that
\begin{equation}
\label{definicja_Omega}
C_{y} = \frac{1}{\Omega_{y}}, \ 2A = \frac{\Omega_{q}}{\Omega_{y}}
\end{equation}
The metric (\ref{metryka_HH_Dee_x_dege_twistfree}) depends on $C_{y}$. Thus we differentiate Eq. (\ref{HH_resztki_1}) with respect to $y$ and using 
 (\ref{definicja_Omega}) we arrive at the equation
\begin{equation}
\label{drugie_rownanie_2}
\Omega_{y} M_{q}-\Omega_{q} M_{y}   = 6 \mu_{0}, \ M:= 2A_{yy}= \partial^{2}_{y} \left(  \frac{\Omega_{q}}{\Omega_{y}} \right)
\end{equation}
Multiplying (\ref{drugie_rownanie_2}) by $dy \wedge dq$ we get
\begin{equation}
\label{drugie_rownanie_3}
d \Omega \wedge  dM  = 6 \mu_{0} \, dy \wedge dq
\end{equation}
Eq. (\ref{drugie_rownanie_3}) suggests that $\Omega$ should be treated as a new variable. Thus, we introduce a coordinate transformation
\begin{equation}
\label{transformacja_coordiantes}
(q,y) \rightarrow (q', w): \ \ q'=q, \ w = \Omega (q,y) \ \ \rightarrow \ \ y = Y (q',w)
\end{equation}
Hence, from (\ref{drugie_rownanie_3}) it follows
\begin{equation}
\label{trzecie_rownanie}
M_{q'} = 6 \mu_{0} Y_{w}
\end{equation}
Transformations of derivatives of an arbitrary function $F=F(q,y) = \widehat{F} (q(q'), Y(q',w)) = \widehat{F} (q',w)$ read
\begin{eqnarray}
\label{transformacja_pochodnych_ffunkci}
F_{y} &=& \Omega_{y} \widehat{F}_{w} 
\\ \nonumber
F_{q} &=& \widehat{F}_{q'} + \Omega_{q} \widehat{F}_{w} 
\end{eqnarray}
Specifying (\ref{transformacja_pochodnych_ffunkci}) for $F = y$ and $\widehat{F} = Y$ one gets
\begin{eqnarray}
\label{transformacja_pochodnych_ffunkci_y}
1 = \Omega_{y} Y_{w} \ &\rightarrow & \ \Omega_{y} = \frac{1}{Y_{w}}
\\ \nonumber
0 = Y_{q'} + \Omega_{q} Y_{w} \ &\rightarrow & \ \Omega_{q} = -\frac{Y_{q'}}{Y_{w}}
\end{eqnarray}
Hence
\begin{equation}
\label{wzory_na_Cy_ina_A}
C_{y} = Y_{w} , \ C_{yy} = \frac{Y_{ww}}{Y_{w}}  , \ 2A = -Y_{q'}
\end{equation}
The only equation which remains to be solved is a consistency condition $M = \partial_{y}^{2} (-Y_{q'})$. Using (\ref{transformacja_pochodnych_ffunkci}) one arrives at the equation
\begin{equation}
\label{czwarte_rownanie}
M Y_{w} = - \partial_{w} \left(  \frac{Y_{q'w}}{Y_{w}}   \right)
\end{equation} 
Eliminating $Y_{w}$ from (\ref{czwarte_rownanie}) by virtue of (\ref{trzecie_rownanie}) we obtain an equation which can be simply integrated twice with respect to $q'$. This proves (\ref{Typ_II_pm_pm_jedyne_rownanie}). To write down the metric we need to express $A$, $A_{yy}$, $C_{y}$ and $C_{yy}$ by $M$:
\begin{equation}
\label{wzory_na_Cy_ina_A_przez_M}
C_{y} = \frac{M_{q'}}{6 \mu_{0}} , \ C_{yy} = \frac{M_{q'w}}{M_{q'}}  , \ 2A = -Y_{q'} , \ 2A_{yy} = M
\end{equation}
Fortunately, factors with $Y_{q'}$ cancel in the metric. Thus, we arrive at (\ref{metryka_TypII_pm_pm_ostateczna}). Finally, after some straightforward but tedious calculations type-[II] conditions (\ref{Typ_II_pm_pm_jedyne_warunki}) are proved.
\end{proof}

A simple example of a solution of Eq. (\ref{Typ_II_pm_pm_jedyne_rownanie}) can be given if we additionally assume that $b(w) =0$ holds true. In this case a general solution of Eq. (\ref{Typ_II_pm_pm_jedyne_rownanie}) can be found and it reads
\begin{equation}
M (q,w) = \left(  \frac{18 \mu_{0} H_{w}}{H+Q} \right)^{\frac{1}{2}}, \ H=H(w), \ Q=Q(q), \ H_{w} \ne 0, \ \partial_{w} \left( \frac{H_{ww}}{H_{w}} \right) \ne 0
\end{equation}
where $H=H(w)$ is a function such that $a(w) = 18 \mu_{0} ( H_{ww} / H_{w})$ .

\subsubsection{Solution with a 2D symmetry algebra}
\label{Typ_II_pm_pm_2D}

In this Section we prove, that the type $\{ [\textrm{D}]^{ee} \otimes [\textrm{II}]^{e}, [+-,+-] \}$ admits at most 2D symmetry algebra. 

\begin{Lemat}
\label{Lemat_o_nieistnieniu_Killinga}
A complex Einstein space of the type $\{ [\textrm{D}]^{ee} \otimes [\textrm{II}]^{e}, [+-,+-] \}$ which metric
is given in the general form (\ref{metryka_TypII_pm_pm_ostateczna}) does not admit any Killing vector other than $K_{1} = \partial_{p}$. 
\end{Lemat}

\begin{proof}
Let us assume that a space of the type $\{ [\textrm{D}]^{ee} \otimes [\textrm{II}]^{e}, [+-,+-] \}$ admits the second
Killing vector $K_{2}$ and assume that for $K_{2}$ we have $\widetilde{a}=0$. From (\ref{zredukowane_master_2}) with $\chi_{0}=0$ it follows that $\widetilde{\epsilon} C_{y} = 0$. Because $C_{y} \ne 0$, we find that $ \widetilde{\epsilon} = 0$. However, in such a case, $K_{2} = 0$ what is a contradiction. Thus, $\widetilde{a} \ne 0$. Hence, without any loss of generality one puts $\widetilde{a}=1$ and $\widetilde{\epsilon}=0$ (compare (\ref{symmetry_transformation})). Consequently, $K_{2} = \partial_{q}$. From Eqs. (\ref{zredukowane_master}) one finds that $C=C(y)$ and $A=A(y)$. Hence, from (\ref{HH_resztki_2}) we get $A = A_{0} C_{y}$ where $A_{0}$ is a constant. Differentiating (\ref{HH_resztki_1}) with respect to $y$ one gets $2A_{0}^{2} C_{yyyy} + 3 \mu_{0}=0$. The condition $\mu_{0} \ne 0$ implies that $A_{0} \ne 0$ and $C_{yyyy} = -\dfrac{3 \mu_{0}}{2 A_{0}^{2}}$. Consequently, $A_{yyy} = - \dfrac{3 \mu_{0}}{2 A_{0}}$. However, with such forms of $A_{yyy}$ and $C_{yyyy}$ we find that both type-[II] conditions are violated what is a contradiction. 
\end{proof}

Now let us assume that the second symmetry is generated by a proper homothetic vector.

\begin{Twierdzenie}
\label{Twierdzenie_typ_II_pm_pm_2Dalgebra}
Let $(\mathcal{M}, ds^{2})$ be an Einstein complex space of the type $\{ [\textrm{D}]^{ee} \otimes [\textrm{II}]^{e}, [+-,+-] \}$ equipped with a 2D symmetry algebra $A_{2,1}$. Then there exists a local
coordinate system $(q,p,x,w)$ such that the metric takes the form (\ref{metryka_TypII_pm_pm_ostateczna}) with $\Lambda=0$ and 
\begin{equation}
\label{rozwiazanie_na_M}
M= e^{ - \frac{2}{3} \chi_{0}q} \, t^{-\frac{1}{2}}S(t), \ \ t:= e^{ - \frac{4}{3} \chi_{0}q} \, w
\end{equation}
where $S=S(t)$ is a holomorphic function such that $S_{t} \ne 0$ and
\begin{equation}
\label{rownanie_Abela_symetria}
36 \mu_{0} t S_{t} = - S^{3} + (a_{0} + 18 \mu_{0}) S + b_{0} 
\end{equation}
$a_{0}$ and $b_{0}$ are constants such that
\begin{equation}
\label{Typ_II_pm_pm_jedyne_warunki_symetria_homotetyczna}
a_{0}  S + \frac{3}{2} b_{0} \ne 0 , \ 3a_{0}  S^{2} + 9 b_{0} S + 72 a_{0} \mu_{0} + a_{0}^{2} \ne 0
\end{equation}
\end{Twierdzenie}

\begin{proof}
A proper homothetic vector can be brought to the form 
\begin{equation}
\label{drugi_homotetyczny_II_pm_pm}
K_{2} =  \partial_{q}  + \frac{2}{3} \chi_{0} (2p \partial_{p} - x \partial_{x} + y \partial_{y} )
\end{equation}
without any loss of generality. Indeed, the assumption $\widetilde{a}=0$ leads to $K_{2}=0$ what is a contradiction (compare Lemma \ref{Lemat_o_nieistnieniu_Killinga}). Thus, $\widetilde{a} \ne 0 \ \Longrightarrow \ \widetilde{a}=1$ and $\widetilde{\epsilon}=0$ (compare (\ref{symmetry_transformation})). Hence, (\ref{drugi_homotetyczny_II_pm_pm}) is proved. The commutation rules read
\begin{equation}
[K_{1}, K_{2}] = \frac{4 \chi_{0}}{3} K_{1}
\end{equation}
Thus, a symmetry algebra is non-abelian $A_{2,1}$. 

With $K_{2}$ in the form (\ref{drugi_homotetyczny_II_pm_pm}), Eqs. (\ref{zredukowane_master}) reads
\begin{subequations}
\label{rownanie_master_dwie_symetrie}
\begin{eqnarray}
\label{rownanie_master_dwie_symetrie_1}
A_{q} + \frac{2}{3} \chi_{0} (y A_{y} - A) = 0
\\ 
\label{rownanie_master_dwie_symetrie_2}
C_{q} + \frac{2}{3} \chi_{0} y C_{y} = 0
\end{eqnarray}
\end{subequations}
Transformation of Eqs. (\ref{rownanie_master_dwie_symetrie}) into the coordinate system $(q',p,x,w)$ (\ref{transformacja_coordiantes}) leads to the equations
\begin{subequations}
\label{rownanie_master_dwie_symetrie_w}
\begin{eqnarray}
\label{rownanie_master_dwie_symetrie_w_1}
\partial_{q'} \left( \frac{Y_{q'}}{Y_{w}} - \frac{2}{3} \chi_{0} \frac{Y}{Y_{w}} \right) =0
\\
\label{rownanie_master_dwie_symetrie_w_2}
\partial_{w} \left( \frac{Y_{q'}}{Y_{w}} - \frac{2}{3} \chi_{0} \frac{Y}{Y_{w}} +  \frac{4}{3} \chi_{0} w \right) =0
\end{eqnarray}
\end{subequations}
what implies
\begin{equation}
\label{rownanie_na_Y}
Y_{q'} - \frac{2}{3}\chi_{0} ( Y - 2wY_{w} ) = 0
\end{equation}
(a constant of integration has been absorbed into $w$). Eq. (\ref{rownanie_na_Y}) can be simply solved and the solution yields $Y = e^{\frac{2}{3} \chi_{0} q'} R(t)$ where $t := e^{-\frac{4}{3} \chi_{0} q'} w$. Hence, 
\begin{equation}
\label{rownanie_na_Mpoqprim}
M_{q'} = 6 \mu_{0} Y_{w} = 6 \mu_{0} e^{-\frac{2}{3} \chi_{0}q'} R_{t}
\end{equation}
Now we introduce a function $S=S(t)$ such that 
\begin{equation}
R_{t} =: - \frac{ 2 \chi_{0}}{9 \mu_{0}} t^{\frac{1}{2}} S_{t}
\end{equation}
It allows to integrate Eq. (\ref{rownanie_na_Mpoqprim}), because
\begin{equation}
\label{rozwiazanie_naM}
M_{q'}=\partial_{q'} (e^{-\frac{2}{3}\chi_{0} q'} t^{- \frac{1}{2} }S ) \ \longrightarrow \ M = e^{-\frac{2}{3}\chi_{0} q'} t^{- \frac{1}{2} } S + m(w)
\end{equation}
Note, that $M_{q'}$ is necessarily nonzero what implies that $S_{t} \ne 0$. Inserting (\ref{rozwiazanie_naM}) into (\ref{Typ_II_pm_pm_jedyne_rownanie}) and expressing the equation in terms of $w$ and $t$ only one obtains
\begin{equation}
\label{przerobione_rownanie}
36 \mu_{0} ( t^{\frac{3}{2}} \partial_{t} (t^{-\frac{1}{2}} S) + w^{\frac{3}{2}} m_{w} ) = -( S + w^{\frac{1}{2}} m)^{3} + w a  S + w^{\frac{3}{2}} (a m + b)
\end{equation}
Differentiating (\ref{przerobione_rownanie}) with respect to $t$ and $w$ one finds
\begin{equation}
a(w) = a_{0} w^{-1}, \ b(w) = b_{0} w^{-\frac{3}{2}}, \ m(w) = m_{0} w^{-\frac{1}{2}}
\end{equation}
with such form of $m$ the constant $m_{0}$ can be incorporated into $S(t)$ what proves (\ref{rozwiazanie_na_M}). Eq. (\ref{przerobione_rownanie}) takes the form (\ref{rownanie_Abela_symetria}). Finally, the conditions (\ref{Typ_II_pm_pm_jedyne_warunki}) reduce to (\ref{Typ_II_pm_pm_jedyne_warunki_symetria_homotetyczna}).
\end{proof}

\begin{Uwaga}
\normalfont
Note, that in this case the Abel equation (\ref{rownanie_Abela_symetria}) is separable and a solution of $t$ in terms of $S$ can be found. The metric (\ref{metryka_TypII_pm_pm_ostateczna}) can be written down explicitly. Denoting $S=u$ and treating $u$ as a new variable instead of $w$ one finds
\begin{eqnarray}
\frac{1}{2} ds^{2} &=& \bigg\{ -dp dx - \frac{ \chi_{0}}{81 \mu_{0}^{2}} e^{-\frac{2\chi_{0}q}{3} } \frac{f}{T}   \left( dp +  e^{\frac{4\chi_{0}q}{3} } T^{2} \left( \frac{2\chi_{0}}{3}   dq + \frac{18 \mu_{0} }{f}  du  \right) dq     \right)
\\
\nonumber
&& \ \ \ \ \ \ + \left(  \mu_{0} x^{3}  -\frac{1}{2} x^{2} e^{-\frac{2\chi_{0}q}{3} } \frac{u}{T}  - \frac{x}{2T^{2}} e^{-\frac{4\chi_{0}q}{3} } \left( 1 - \frac{3u^{2} - a_{0} +18 \mu_{0}}{18 \mu_{0}}  \right)  \right) dp^{2} \bigg\}
\end{eqnarray}
where
\begin{equation}
T(u) = \exp \left( 18 \mu_{0} \int \frac{du}{f(u)} \right), \ f(u) := -u^{3} + (18 \mu_{0} + a_{0}) u + b_{0}
\end{equation}

\end{Uwaga}

A symmetry algebra with more then 2 dimensions is not possible within the type $\{ [\textrm{D}]^{ee} \otimes [\textrm{II}]^{e}, [+-,+-] \}$. Indeed, from Lemma \ref{Lemat_o_nieistnieniu_Killinga} it follows that no Killing vector other then $K_{1} = \partial_{p}$ is admitted and an arbitrary space admits only one proper homothetic vector.

\subsection{Type $\{ [\textrm{D}]^{ee} \otimes [\textrm{II}]^{e}, [+-,--] \}$}
\label{Sekcja_typ_II_pm_mm}

In this Section we deal with the case for which
\begin{equation}
\label{warunki_Typ_II_pm_mm}
C_{y} = 0,  \ A_{yyy} \ne 0 
\end{equation}

\subsubsection{General case}

\begin{Twierdzenie}
\label{Twierdzenie_typ_II_pm_mm}
Let $(\mathcal{M}, ds^{2})$ be an Einstein complex space of the type $\{ [\textrm{D}]^{ee} \otimes [\textrm{II}]^{e}, [+-,--] \}$. Then there exists a local coordinate system $(q,p,x,w)$ such that the metric takes the form 
\begin{equation}
\label{metryka_TypII_pm_mm_ostateczna}
\frac{1}{2} ds^{2} = x^{-2} \bigg\{ -dpdx + \left( \mu_{0} x^{3} + \frac{\Lambda}{6} - w x^{2}  \right) dp^{2} - \frac{F_{w}}{w (q+F)^{2}} \, dq dw \bigg\}
\end{equation}
where $\mu_{0}=1$, $\Lambda$ is the cosmological constant, $F=F(w)$ is an arbitrary holomorphic function such that $F_{w} \ne 0$. 
\end{Twierdzenie}

\begin{proof}
If $C_{y}=0$ then Eq. (\ref{HH_resztki_2}) is identically satisfied. Differentiating Eq. (\ref{HH_resztki_1}) with respect to $y$ and multiplying the result by $dy \wedge dq$ one gets
\begin{equation}
\label{rrrownanie_2}
-2A \, dA_{yy} \wedge dq + dy \wedge dA_{yy} = 0
\end{equation}
Analogously like in the proof of Theorem \ref{Twierdzenie_typ_II_pm_pm} we change the coordinates according to the formulas
\begin{equation}
\label{transformacja_coordiantes_prostrzy_II}
(q,y) \rightarrow (q', w): \ \ q'=q, \ w = A_{yy} (q,y) \ \ \rightarrow \ \ y = Y (q',w)
\end{equation}
Hence, from (\ref{rrrownanie_2}) one finds
\begin{equation}
\label{wzor_na_A}
2A = - Y_{q'}
\end{equation}
The consistency condition between (\ref{wzor_na_A}) and $w=A_{yy}$ gives
\begin{equation}
\label{roownanie_na_Y}
2wY_{w} = - \partial_{w} \left(  \frac{Y_{q'w}}{Y_{w}}   \right) 
\end{equation}
Eq. (\ref{roownanie_na_Y}) is the Liouville differential equation with a general solution
\begin{equation}
\label{rozwiazanie_na_Y}
Y_{w} = - \frac{F_{w} Q_{q'}}{w (Q+F)^{2}} , \ F=F(w), \ Q=Q(q')
\end{equation}
The metric depends on $A$ which is given in terms of $Y_{q'}$. Fortunately, terms with $Y_{q'}$ cancel in the metric. Thus, $Y_{w}$ is all we need to write down the metric. Moreover, in the metric a factor $Q_{q'} dq' = dQ$ appears. Denoting $Q$ by $q$ we arrive at (\ref{metryka_TypII_pm_mm_ostateczna}). Type-[II] condition yields $F_{w} \ne 0$. 
\end{proof}

\subsubsection{Solution with a 2D symmetry algebra}

To prove that a space of the type $\{ [\textrm{D}]^{ee} \otimes [\textrm{II}]^{e}, [+-,--] \}$ admits at most 2D algebra of infinitesimal symmetries we follow the pattern from Section \ref{Typ_II_pm_pm_2D}. First we prove that no other Killing vectors except $K_{1} = \partial_{p}$ are admitted and then we consider a case which admits a proper homothetic vector.
\begin{Lemat}
\label{Lemat_o_nieistnieniu_Killinga_typ_II_pmmm}
A complex Einstein space of the type $\{ [\textrm{D}]^{ee} \otimes [\textrm{II}]^{e}, [+-,--] \}$ which metric
is given in the general form (\ref{metryka_TypII_pm_mm_ostateczna}) does not admit any Killing vector other than $K_{1} = \partial_{p}$. 
\end{Lemat}

\begin{proof}
Assume first that the second homothetic vector $K_{2}$ exists. Assume also that $\widetilde{a}=0$. Then from (\ref{zredukowane_master_1}) one finds that $\chi_{0} = \widetilde{\epsilon}=0$ (what implies that $K_{2}=0$) or $A_{yy}=0$ (what implies that a space is not type-[II] anymore). Both these possibilities are contradictions. Thus, $\widetilde{a} \ne 0$ holds true and without any loss of generality one puts $\widetilde{a}=1$ and $\widetilde{\epsilon}=0$ (compare (\ref{symmetry_transformation})). Hence, $K_{2}$ has the form (\ref{drugi_homotetyczny_II_pm_pm}). If $\chi_{0} =0$ then Eq. (\ref{zredukowane_master_1}) yields $A_{q}=0$ and from the field equation (\ref{HH_resztki_1}) we find that $A_{yyy} =0$ what is a contradiction. Thus, $\chi_{0} \ne 0$ and $K_{2}$ is a proper homothetic vector.
\end{proof}

\begin{Twierdzenie}
Let $(\mathcal{M}, ds^{2})$ be an Einstein complex space of the type $\{ [\textrm{D}]^{ee} \otimes [\textrm{II}]^{e}, [+-,--] \}$ equipped with a 2D symmetry algebra $A_{2,1}$. Then there exists a local
coordinate system $(q,p,x,w)$ such that the metric takes the form (\ref{metryka_TypII_pm_mm_ostateczna}) with $\Lambda=0$ and 
\begin{equation}
\label{rozwiazanie_na_F}
F(w) = \frac{3}{2 \chi_{0}} \ln w
\end{equation}
\end{Twierdzenie}

\begin{proof}
From Lemma \ref{Lemat_o_nieistnieniu_Killinga_typ_II_pmmm} it follows that the second symmetry is generated by the proper homothetic vector $K_{2}=\partial_{q}  + \frac{2}{3} \chi_{0} (2p \partial_{p} - x \partial_{x} + y \partial_{y} )$. The commutation rule reads $[K_{1}, K_{2}] = \frac{4 \chi_{0}}{3} K_{1}$ what implies that a symmetry algebra is 2D non-abelian $A_{2,1}$. Eq. (\ref{zredukowane_master_1}) takes the form (\ref{rownanie_master_dwie_symetrie_1}) or - in the coordinate system $(q,p,x,w)$ - (\ref{rownanie_master_dwie_symetrie_w_1}). It yields
\begin{equation}
\label{dodatkow_na_Y}
Y_{q'} - \frac{2}{3}\chi_{0} Y = \alpha(w) \, Y_{w}
\end{equation}
where $\alpha(w)$ is an arbitrary function. Differentiating (\ref{dodatkow_na_Y}) with respect to $w$ and feeding the result with (\ref{rozwiazanie_na_Y}) (remember that $Q=q'$) one finds $\alpha F_{w}=1$ and, consequently
\begin{equation}
F(w) = \frac{3}{2 \chi_{0}} \ln w + F_{0}
\end{equation}
The constant $F_{0}$ can be absorbed into $q$. Hence, (\ref{rozwiazanie_na_F}) is proved. 
\end{proof}


\section{Types $[\textrm{D}]^{ee} \otimes [\textrm{D}]^{ee}$}
\label{section_Dee_x_Dee}
\setcounter{equation}{0}

In this Section we focus on the case for which type-[D] conditions are satisfied
\begin{subequations}
\label{type_D_roownanie}
\begin{eqnarray}
\label{type_D_roownanie_1}
&& A_{yyyy} = 0
\\
\label{type_D_roownanie_2}
&& 2A_{yyy}^{2} +3 \mu_{0} C_{yyyy} =0
\end{eqnarray}
\end{subequations}

\begin{Lemat}
The key function $W$ which satisfies type-[D] conditions (\ref{type_D_roownanie}) and field equations (\ref{HH_resztki}) has the form
\begin{eqnarray}
\label{W_dla_typu_D_eexee}
W &=& (a y^{3} + b y^{2} + d y + e) \, x  
\\ \nonumber
&& +\frac{1}{3 \mu_{0}} (-3 a^{2} y^{4} -4ab y^{3} +(3 a_{q} -6ad) y^{2} +(2b_{q} -12 ae) y -4be + d^{2} + d_{q})
\end{eqnarray}
where $\mu_{0}=1$; $a$, $b$, $d$ and $e$ are functions of $q$ only, which satisfy the equations
\begin{subequations}
\label{zredukowane_rownania_pola_typ_D}
\begin{eqnarray}
\partial_{q} (2b^{2} + 3 a_{q} - 6ad) = 0
\\
b_{qq} - 12ea_{q} - 6ae_{q} + 2d b_{q} = 0
\end{eqnarray}
\end{subequations}
\end{Lemat}
\begin{proof}
From (\ref{type_D_roownanie_1}) we find
\begin{equation}
\label{postac_A_typ_D}
A (q,y) = a y^{3} + b y^{2} + d y + e
\end{equation}
where $a$, $b$, $d$ and $e$ are arbitrary functions of $q$ only. From the field equation (\ref{HH_resztki_1}) we find the form of $C$
\begin{equation}
\label{postac_C_typ_D}
3 \mu_{0} C = -3 a^{2} y^{4} -4ab y^{3} +(3 a_{q} -6ad) y^{2} +(2b_{q} -12 ae) y -4be + d^{2} + d_{q}
\end{equation}
Eq. (\ref{type_D_roownanie_2}) is identically satisfied by virtue of (\ref{postac_A_typ_D}) and (\ref{postac_C_typ_D}). After some straightforward calculations one finds that the field equation (\ref{HH_resztki_2}) reduces to the system (\ref{zredukowane_rownania_pola_typ_D}). Finally, inserting (\ref{postac_A_typ_D}) and (\ref{postac_C_typ_D}) into (\ref{funkcja_kluczowa_postac_wyjsciowa_uproszczona}), the form (\ref{W_dla_typu_D_eexee}) is proved.
\end{proof}
Under (\ref{gauge}) functions $a$, $b$, $d$ and $e$ transform as follows
\begin{subequations}
\label{transformacje_na_funkcje_abcd}
\begin{eqnarray}
\label{transformacje_na_funkcje_abcd_a}
a' &=& f a
\\ 
\label{transformacje_na_funkcje_abcd_b}
b' &=& b-3a' \sigma
\\ 
\label{transformacje_na_funkcje_abcd_d}
d' &=& \frac{d}{f} + \frac{f_{q}}{2 f^{2}} - 3a' \sigma^{2} - 2b' \sigma
\\ 
\label{transformacje_na_funkcje_abcd_e}
e' &=& \frac{e}{f^{2}} - \frac{\sigma_{q}}{2f} - a' \sigma^{3}-b' \sigma^{2} - d' \sigma
\end{eqnarray}
\end{subequations}

Before we solve Eqs. (\ref{zredukowane_rownania_pola_typ_D}), we establish the properties of additional ASD $\mathcal{C}$ and $\mathcal{I}s$ which automatically appear when the ASD Weyl tensor is type-[D] (See Figure \ref{Congruences_4}).

\begin{figure}[ht]
\begin{center}
\includegraphics[scale=0.9]{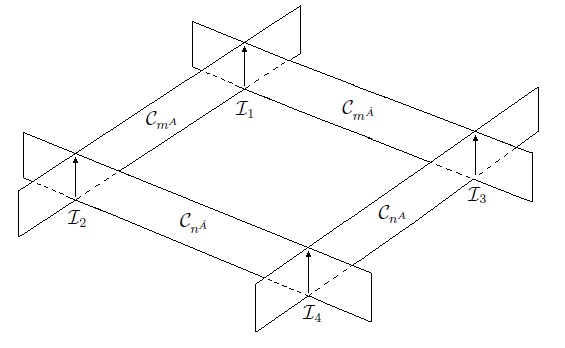}
\caption {Congruences of null strings and congruences of null geodesics in type-[D] para-Hermite Einstein spaces.}
\label{Congruences_4}
\end{center}
\end{figure}

The ASD Weyl spinor (\ref{spinorowa_postac_ASD_Weyl_1}) takes the form 
\begin{equation}
\label{spinorowa_postac_ASD_Weyl_2}
C_{\dot{A}\dot{B}\dot{C}\dot{D}} = -6 \mu_{0} x^{3} \, n_{(\dot{A}} n_{\dot{B}} k_{\dot{C}} k_{\dot{D})}, \ n_{\dot{A}} = l_{\dot{A}} - \frac{2a}{\mu_{0}} k_{\dot{A}}
\end{equation}
If we chose a spinorial basis is such a manner that $l_{\dot{A}} = [i,0]$, $k_{\dot{A}} = [0, i]$, the second double Penrose spinor $n_{\dot{A}}$ reads
\begin{equation}
n_{\dot{A}} = i \left[ 1, - \frac{2a}{\mu_{0}} \right]
\end{equation}
The expansion of $\mathcal{C}_{n^{\dot{A}}}$ is given by the formulas\footnote{The expansion of an ASD congruence of null strings generated by a spinor $n_{\dot{A}}$ such that $n_{\dot{A}} \sim [1, w]$ was calculated in \cite{Chudecki_Ref_2}. It is enough to write formulas (2.38a) and (2.38b) of \cite{Chudecki_Ref_2} for $\phi = x$ and $w = -2a / \mu_{0}$ to prove that (\ref{expansja_czwartej_struny}) is correct.}
\begin{equation}
\label{expansja_czwartej_struny}
N_{1} = \frac{\sqrt{2}}{x} \frac{2a}{\mu_{0}}, \ N_{2} = \sqrt{2} x \left( \mathcal{A} - \frac{2a}{\mu_{0}} \mathcal{Q} \right)
\end{equation}
Because $\mu_{0} \ne 0$ then $N_{2}$ is always nonzero. Thus, $\mathcal{C}_{n^{\dot{A}}}$ is always expanding. 

From (\ref{wlasnosci_przeciec_po_zmianach_I3}) it follows that $\varrho_{3} = 0$ and $\theta_{3} \sim C_{y}$. However, $C_{y} = 0 \Longleftrightarrow a=0$ (it follows from (\ref{postac_C_typ_D}) and (\ref{zredukowane_rownania_pola_typ_D})). Hence
\begin{eqnarray}
\label{wlasnosci_typ_D_I3}
p(\mathcal{I}_{3}) = [+-] \ &\textrm{for}& \ a \ne 0
\\ \nonumber
p(\mathcal{I}_{3}) = [--] \ &\textrm{for}& \ a = 0
\end{eqnarray}
Properties of $\mathcal{I}_{2}$ and $\mathcal{I}_{4}$ are given by the formulas
\begin{eqnarray}
&& \theta_{2}  \sim m_{A} N^{A} + n_{\dot{A}} M^{\dot{A}}, \ \ \varrho_{2} \sim m_{A} N^{A} - n_{\dot{A}} M^{\dot{A}}
\\ \nonumber
&& \theta_{4} \sim n_{A} N^{A} + n_{\dot{A}} N^{\dot{A}}, \ \ \varrho_{4} \sim n_{A} N^{A} - n_{\dot{A}} N^{\dot{A}}
\end{eqnarray}
Hence, we find that 
\begin{eqnarray}
\label{wlasnosci_typ_D_I2}
p(\mathcal{I}_{2}) = [+-] \ &\textrm{for} \ a \ne 0&
\\ \nonumber
p(\mathcal{I}_{2}) = [--] \ &\textrm{for} \ a = 0 &
\\ \nonumber
p(\mathcal{I}_{4}) = [+-] \ &\textrm{for} \ \textrm{an arbitrary } a &
\end{eqnarray}
Thus, we deal here with the types
\begin{eqnarray}
\nonumber
\{ [\textrm{D}]^{ee} \otimes [\textrm{D}]^{ee},[+-,+-,+-,+-] \} &\textrm{ for }& a \ne 0
\\ \nonumber
\{ [\textrm{D}]^{ee} \otimes [\textrm{D}]^{ee},[+-,--,--,+-] \} &\textrm{ for }& a = 0
\end{eqnarray}

\begin{Twierdzenie}
Let $(\mathcal{M}, ds^{2})$ be an Einstein complex space of the type $\{ [\textrm{D}]^{ee} \otimes [\textrm{D}]^{ee}, [+-,+-,+-,+-]\}$. Then there exists a local coordinate system $(q,p,x,y)$ such that the metric takes the form 
\begin{eqnarray}
\label{metryka_TypD_ostateczna_pmpmpmpm}
\frac{1}{2} ds^{2} &=& x^{-2} \bigg\{  dq dy - dp dx + \left(  \mu_{0} x^{3} +\frac{\Lambda}{6} -6yx^{2} + \frac{4}{\mu_{0}} x (3y^{2} + d_{0})    \right)   dp^{2}
\\ \nonumber
&& \ \ \ \ \ \ - \frac{8}{\mu_{0}} ( y^{3} +d_{0} y + e_{0}  ) 
\,  dpdq  + 2 (y^{3} + d_{0}y + e_{0}) \, dq^{2} \bigg\}
\end{eqnarray}
where $\mu_{0} = 1$, $\Lambda$ is the cosmological constant, $d_{0}$ and $e_{0}$ are constants. If $d_{0} \ne 0$ or $e_{0} \ne 0$ or $\Lambda \ne 0$ then the metric (\ref{metryka_TypD_ostateczna_pmpmpmpm}) is equipped with a 2D  symmetry algebra $2A_{1}$ with generators
\begin{equation}
\label{symmetry_algebra_type_D_2D}
K_{1} = \partial_{p}, \ K_{2} = \partial_{q}, \ [K_{1}, K_{2}]=0
\end{equation}
If $d_{0}=e_{0}=\Lambda=0$ then the metric (\ref{metryka_TypD_ostateczna_pmpmpmpm}) is equipped with a 3D  symmetry algebra $A_{3,3}$ with generators
\begin{eqnarray}
\label{symmetry_algebra_type_D_3D}
&& K_{1} = \partial_{p}, \ K_{2} = \partial_{q}, \ K_{3} = \frac{2}{3}\chi_{0} \, (2p \partial_{p} - x \partial_{x} + 2q \partial_{q} - y \partial_{y} )
\\ \nonumber
&& [K_{1}, K_{2}]=0, \ [K_{1}, K_{3}] = \frac{4}{3} \chi_{0} \, K_{1}, \ [K_{2}, K_{3}] = \frac{4}{3} \chi_{0} \, K_{2} 
\end{eqnarray}
\end{Twierdzenie}

\begin{proof}
For this class of type-[D] metrics the function $a$ is necessarily nonzero. Thus, it can be brought to the constant value, say, $a=1$, without any loss of generality (compare (\ref{transformacje_na_funkcje_abcd_a})). Similarly, one puts $b=0$ (compare (\ref{transformacje_na_funkcje_abcd_b}). From Eqs. (\ref{zredukowane_rownania_pola_typ_D}) we get $d=d_{0} = \textrm{const}$ and  $e=e_{0} = \textrm{const}$. Consequently,
\begin{equation}
\label{W_dla_typu_D_eexee_ane0}
W = (y^{3}  + d_{0} y + e_{0}) \, x    -  \frac{1}{3 \mu_{0}} (3 y^{4}   +6d_{0} y^{2}  +12 e_{0} y  - d_{0}^{2} )
\end{equation}
and the metric (\ref{metryka_TypD_ostateczna_pmpmpmpm}) is proved.

From Eqs. (\ref{zredukowane_master}) we find that 
\begin{equation}
\label{warunki_na_symetrie_typ_D_pmpmpmpm}
\widetilde{\epsilon}=0, \ \chi_{0} d_{0} = 0 , \ \chi_{0} e_{0} = 0, \ \chi_{0} \Lambda=0, \  \widetilde{a} = \frac{4}{3} \chi_{0} q + \widetilde{a}_{0}
\end{equation}
Thus, any homothetic vector other then $K_{1} = \partial_{p}$ has the form
\begin{equation}
\label{dodatkowy_homotetyczny_dla_typD}
K = \widetilde{a}_{0} \partial_{q} + \frac{2}{3} \chi_{0} \, (2p \partial_{p} -x \partial_{x} + 2q \partial_{q} - y \partial_{y} )
\end{equation}
From (\ref{dodatkowy_homotetyczny_dla_typD}) and (\ref{warunki_na_symetrie_typ_D_pmpmpmpm}) we conclude that the second Killing vector $K_{2} = \partial_{q}$ always exists, but the existence of the third, proper homothetic vector $K_{3} = \frac{2}{3} \chi_{0} \, (2p \partial_{p} -x \partial_{x} + 2q \partial_{q} - y \partial_{y} )$ depends on the conditions (\ref{warunki_na_symetrie_typ_D_pmpmpmpm}). If at least one of the constants $d_{0}$, $e_{0}$ or $\Lambda$ is nonzero, then $\chi_{0}=0$ and the symmetry algebra is 2-dimensional abelian $2A_{1}$. Thus, (\ref{symmetry_algebra_type_D_2D}) is proved. However, if $d_{0}=e_{0}=\Lambda=0$ then $K_{3}$ exists and the symmetry algebra is 3-dimensional. Substituting
\begin{equation}
e_{1} := K_{1}, \ e_{2} := K_{2}, \ e_{3} := \frac{4}{3}\chi_{0} K_{3}
\end{equation}
one finds the commutation rules
\begin{equation}
[e_1, e_2]=0 , \ [e_{1}, e_{3}]=e_1, \ [e_2, e_3]=e_2
\end{equation}
From Table I of \cite{Patera} we find that such commutation rules correspond to the algebra $A_{3,3}$.
\end{proof}

\begin{Twierdzenie}
Let $(\mathcal{M}, ds^{2})$ be an Einstein complex space of the type $\{ [\textrm{D}]^{ee} \otimes [\textrm{D}]^{ee}, [+-,--,--,+-]\}$. Then there exists a local coordinate system $(q,p,x,y)$ such that the metric takes the form 
\begin{eqnarray}
\label{metryka_TypD_ostateczna_pmmmmmpm}
\frac{1}{2} ds^{2} &=& x^{-2} \bigg\{  dq dy - dp dx + \left(  \mu_{0} x^{3} +\frac{\Lambda}{6} -2b_{0} x^{2}   \right)   dp^{2} + 2b_{0} y^{2} \, dq^{2}  \bigg\}
\end{eqnarray}
where $\mu_{0} = 1$, $\Lambda$ is the cosmological constant, $b_{0}$ is a constant. If $b_{0} \ne 0$ then the metric (\ref{metryka_TypD_ostateczna_pmmmmmpm}) is equipped with a 4D symmetry algebra $A_{3,8} \oplus A_{1}$ with generators
\begin{eqnarray}
\label{symmetry_algebra_type_D_2D4}
&& K_{1} = \partial_{p}, \ K_{2} = \partial_{q}, \ K_{3} = q \partial_{q} - y \partial_{y}, \ K_{4} = \partial_{y} + 2b_{0}q (q \partial_{q} - 2y \partial_{y} ),
\\ \nonumber
&& [K_{1}, K_{2}]=0, \ [K_1, K_3]=0, \ [K_1, K_4]=0, \ 
\\ \nonumber
&& [K_2, K_3]=K_2, \ [K_2, K_4] = 4 b_{0} K_3, \ [K_3, K_4]=K_4
\end{eqnarray}
If $b_{0} =0$ and $\Lambda \ne 0$ then the metric (\ref{metryka_TypD_ostateczna_pmmmmmpm}) is equipped with a 4D symmetry algebra $A_{3,4} \oplus A_{1}$ with generators (\ref{symmetry_algebra_type_D_2D4}) written for $b_{0}=0$.

If $b_{0} = 0$ and $\Lambda =0$ then the metric (\ref{metryka_TypD_ostateczna_pmmmmmpm}) is equipped with a 5D symmetry algebra $A_{5,33}^{\frac{1}{2},-1}$ with generators
\begin{eqnarray}
\label{symmetry_algebra_type_D_2D5}
&& K_{1} = \partial_{p}, \ K_{2} = \partial_{q}, \ K_{3} = q \partial_{q} - y \partial_{y}, \ K_{4} = \partial_{y}, \ K_{5} = \frac{2}{3} \chi_{0} (2p \partial_{p} - x \partial_{x} + y \partial_{y}), \ \ \ \ \
\\ \nonumber
&& [K_{1}, K_{2}]=0, \ [K_1, K_3]=0, \ [K_1, K_4]=0, \ [K_1, K_5]=\frac{4}{3} \chi_{0} K_{1},
\\ \nonumber
&& [K_2, K_3]=K_2, \ [K_2, K_4] = 0, \ [K_2, K_5]=0
\\ \nonumber
&& [K_3, K_4]=K_4, \ [K_3, K_5]=0, \ [K_4, K_5] = \frac{2}{3} \chi_{0} K_4
\end{eqnarray}
\end{Twierdzenie}

\begin{proof}
If $a=0$ then both $d$ and $e$ can be gauged away without any loss of generality (compare (\ref{transformacje_na_funkcje_abcd_d})-(\ref{transformacje_na_funkcje_abcd_e})). From Eqs. (\ref{zredukowane_rownania_pola_typ_D}) one finds that $b=b_{0}=\textrm{const}$. Hence, $A=b_{0} y^{2}$ and $C=0$. The key function takes very simple form
\begin{equation}
\label{W_dla_typu_D_eexee_prostrzy}
W =  b_{0} \, y^{2} x  
\end{equation}
It proves (\ref{metryka_TypD_ostateczna_pmmmmmpm}). From Eqs. (\ref{zredukowane_master}) we get the following conditions
\begin{equation}
\widetilde{\epsilon} = \widetilde{\epsilon}_{0} = \textrm{const}, \ b_{0} \chi_{0} = 0, \ \chi_{0} \Lambda=0 , \ \widetilde{a} = -2 b_{0} \widetilde{\epsilon}_{0} q^{2} + m_{0} q + n_{0}
\end{equation}
Thus, any homothetic vector admitted by the metric (\ref{metryka_TypD_ostateczna_pmmmmmpm}) other then $K_{1} = \partial_{p}$ has the form
\begin{equation}
K=n_{0} \partial_{q} + m_{0} (q \partial_{q} -  y \partial_{y}) - \widetilde{\epsilon}_{0} ( 2b_{0} q^{2} \partial_{q} - 4 b_{0} qy \partial_{y} + \partial_{y} ) + \frac{2}{3} \chi_{0} (2p \partial_{p} -x \partial_{x} +y \partial_{y})
\end{equation}
If $b_{0} \ne 0$ then $\chi_{0}=0$ and three additional Killing vectors exist. It proves (\ref{symmetry_algebra_type_D_2D4}). Substituting
\begin{equation}
e_0 := K_1,\ e_1 := \frac{1}{4 b_{0}} K_2, \ e_2 := K_3, \ e_3 := -2 K_4
\end{equation}
one finds that $[e_{0}, e_i]=0$ for $i=1,2,3$ and 
\begin{equation}
[e_1, e_2]=e_1, \ [e_1, e_3] = -2 e_2, \ [e_2, e_3]=e_3
\end{equation}
Hence, a symmetry algebra is 4-dimensional $A_{3,8} \oplus A_{1}$.

If $b_{0}=0$ but $\Lambda \ne 0$, then still $\chi_{0}=0$ and the substitution
\begin{equation}
e_0 := K_1,\ e_1 :=  K_2, \ e_2 := K_4, \ e_3 := K_3
\end{equation}
leads to the commutation rules $[e_{0}, e_i]=0$ for $i=1,2,3$ and 
\begin{equation}
[e_1, e_2]=0, \ [e_1, e_3] =  e_1, \ [e_2, e_3]=-e_2
\end{equation}
A symmetry algebra is then $A_{3,4} \oplus A_1$. 

However, if $b_{0}=\Lambda=0$ then the fifth symmetry exists and it is generated by a proper homothetic vector. A symmetry algebra is 5-dimensional. Substituting 
\begin{equation}
e_1 := K_2, \ e_2 := K_1, \ e_3 := K_4, \ e_4 := K_3, \ e_5 := \frac{3}{4 \chi_{0}} K_5
\end{equation}
one arrives at the commutation rules
\begin{eqnarray}
&& [e_{1}, e_{2}]=0, \ [e_1, e_3]=0, \ [e_1, e_4]=e_1, \ [e_1, e_5]=0,
\\ \nonumber
&& [e_2, e_3]=0, \ [e_2, e_4] = 0, \ [e_2, e_5]=e_2
\\ \nonumber
&& [e_3, e_4]=-e_3, \ [e_3, e_5]=\frac{1}{2} e_3, \ [e_4, e_5] = 0
\end{eqnarray}
According to Table II of \cite{Patera} these are exactly the commutation rules for the algebra $A_{5,33}^{ab}$ with $a=\frac{1}{2}$ and $b=-1$.
\end{proof}

\begin{Uwaga}
\normalfont
Note that the Killing vector $K_{4} = \partial_{y}$ is null so the metric (\ref{metryka_TypD_ostateczna_pmmmmmpm}) is equipped with a null Killing vector. Moreover, in the case with $\Lambda=b_{0}=0$ the metric (\ref{metryka_TypD_ostateczna_pmmmmmpm}) is quipped with a proper homothetic vector $K_{5} = \frac{2}{3} \chi_{0} (2p \partial_{p} - x \partial_{x} + y \partial_{y}$. Thus, the metric (\ref{metryka_TypD_ostateczna_pmmmmmpm}) with $\Lambda=b_{0}=0$ is an interesting example of a metric which is simultaneously equipped with a proper homothetic and a null Killing vectors. 
\end{Uwaga}


\section{Concluding remarks}
\setcounter{equation}{0}

The current paper is a continuation of the article \cite{Chudecki_Ref_3}. We analyzed here a special class of algebraically degenerate para-Hermite Einstein spaces. The class is characterized by two properties. The first property is that an ASD congruence of null strings is expanding what implies that the spaces analyzed in this paper are not Walker spaces. The second property is that all congruences of null geodesics which are intersections of SD and ASD congruences of null strings are nontwisting. An interesting fact is that the assumption that a single congruence of null geodesics is nontwisting implies that all other congruences of null geodesics are also nontwisting. Equivalently, all 3D-distributions spanned by SD and ASD $\mathcal{C}s$ are integrable (compare Remark \ref{Uwaga_o_calkowalnosci}). We proved that the only pHE-spaces which satisfy these two conditions are spaces for which ASD Weyl tensor is type-[II] or type-[D]. 

Within this class vacuum Einstein field equations with cosmological constant can be integrated, leading to the explicit metrics. Only the most general type $\{ [\textrm{D}]^{ee} \otimes [\textrm{II}]^{e} , [+-,+-]\}$ has not been completely solved although the field equations have been reduced to the Abel equation of the first kind in this case. The results are gathered in the Table \ref{summary}.

\begin{table}[ht]
 \footnotesize
\begin{center}
\begin{tabular}{|c|c|c|}   \hline
 Type  &   Metric   & Functions in the metric        \\ \hline \hline
 $\{ [\textrm{D}]^{ee} \otimes [\textrm{II}]^{e} , [+-,+-]\}$ & (\ref{metryka_TypII_pm_pm_ostateczna}) & 1 function of 2 variables restricted by Eq. (\ref{Typ_II_pm_pm_jedyne_rownanie}), $\Lambda$   \\ \hline
 $\{ [\textrm{D}]^{ee} \otimes [\textrm{II}]^{e} , [+-,--]\}$ & (\ref{metryka_TypII_pm_mm_ostateczna})  & 1 function of 1 variable, $\Lambda$   \\ \hline
  $\{ [\textrm{D}]^{ee} \otimes [\textrm{D}]^{ee} , [+-,+-,+-,+-]\}$ & (\ref{metryka_TypD_ostateczna_pmpmpmpm})  &  2 constants, $\Lambda$  \\ \hline
  $\{ [\textrm{D}]^{ee} \otimes [\textrm{D}]^{ee} , [+-,--,--,+-]\}$ & (\ref{metryka_TypD_ostateczna_pmmmmmpm})  &   1 constant, $\Lambda$ \\ \hline
\end{tabular}
\caption{Summary of main results.}
\label{summary}
\end{center}
\end{table}

The only class of algebraically degenerate pHE-spaces which remains to be solved is a class for which twists of all congruences of null geodesics are nonzero. It means that these congruences must be also expanding. This case is much more complicated and preliminary analysis indicates that within this class general solutions probably cannot be (relatively) easily obtained. Nevertheless, the problem is now analyzed.

\begin{table}[ht]
 \footnotesize
\begin{center}
\begin{tabular}{|l|l|}   \hline
 Type     &  Notes      \\ \hline \hline
 $ [\textrm{D}]^{ee} \otimes [\textrm{I}]$  &  Algebraically general pHE-spaces, no examples are known \\  \hline
 $\{ [\textrm{D}]^{ee} \otimes [\textrm{II}]^{n} , [++,++]\}$   & Algebraically special pHE-spaces, ASD congruences of null   \\
 $\{ [\textrm{D}]^{ee} \otimes [\textrm{II}]^{n} , [++,--]\}$   &  strings are nonexpanding (spaces are Walker spaces). \\
 $\{ [\textrm{D}]^{ee} \otimes [\textrm{D}]^{nn} , [++,++,++,++]\}$   &   All found in  \cite{Chudecki_Ref_3}. \\
 $\{ [\textrm{D}]^{ee} \otimes [\textrm{D}]^{nn} , [++,--,--,++]\}$   &   \\
 $\{ [\textrm{D}]^{ee} \otimes [\textrm{III}]^{n} , [++,++]\}$    &    \\
 $\{ [\textrm{D}]^{ee} \otimes [\textrm{III}]^{n} , [++,--]\}$    &    \\
 $\{ [\textrm{D}]^{ee} \otimes [\textrm{N}]^{n} , [++,++]\}$   &    \\   \hline
 $\{ [\textrm{D}]^{ee} \otimes [\textrm{II}]^{e} , [+-,+-]\}$  &  Algebraically special pHE-spaces, ASD congruences of null \\ 
 $\{ [\textrm{D}]^{ee} \otimes [\textrm{II}]^{e} , [+-,--]\}$   &   strings are expanding (spaces are not Walker spaces),  \\ 
  $\{ [\textrm{D}]^{ee} \otimes [\textrm{D}]^{ee} , [+-,+-,+-,+-]\}$   &   congruences of null geodesics are all nontwisting. \\ 
  $\{ [\textrm{D}]^{ee} \otimes [\textrm{D}]^{ee} , [+-,--,--,+-]\}$   &   All found in the current paper.  \\ \hline
  $\{ [\textrm{D}]^{ee} \otimes [\textrm{II}]^{e} , [++,++]\}$      & Algebraically special pHE-spaces, ASD congruences of null   \\
  $\{ [\textrm{D}]^{ee} \otimes [\textrm{D}]^{ee} , [++,++,++,++]\}$   &  strings are expanding (spaces are not Walker spaces),  \\
  $\{ [\textrm{D}]^{ee} \otimes [\textrm{III}]^{e} , [++,++]\}$      & congruences of null geodesics are all twisting and expanding.   \\
  $\{ [\textrm{D}]^{ee} \otimes [\textrm{N}]^{e} , [++,++]\}$      &  Work in progress.   \\  \hline
\end{tabular}
\caption{Types of pHE-spaces.}
\label{landscape}
\end{center}
\end{table}

We will close this paper with a brief overview of the types of pHE-spaces. They are gathered in the Table \ref{landscape}. PHE-spaces can be divided into four disjoint families, two of them are solved, two of them are not.

\end{document}